\documentclass[a4paper,10pt]{article}
 \usepackage[english]{babel}
  \usepackage{amsmath,amsfonts,amssymb,amsthm}
  \usepackage{graphicx,xcolor}
\usepackage{mathptmx}
\usepackage{latexsym}
\usepackage{mathrsfs}
\usepackage[T1]{fontenc}
\usepackage[utf8]{inputenc}
\usepackage{authblk}

\newtheorem{theorem}{Theorem}[section]

\newtheorem{lemma}[theorem]{Lemma}
\newtheorem{proposition}[theorem]{Proposition}

\theoremstyle{definition}

\def\Z{\mathbb{Z}}

\let\.=\cdot
\let\0=\emptyset

\def\cst{{\rm const}}
\renewcommand{\a}{\alpha}

\newcommand{\be}{\begin{equation}}
\newcommand{\ee}{\end{equation}}
\newcommand{\ba}{\begin{eqnarray}}
\newcommand{\ea}{\end{eqnarray}}

\def\square{\hbox{$\sqcap\kern-7pt\sqcup$}}

\def\P0{\Psi_{\{0\}}}

\def\cT{\mathcal{T}}
\def\cC{\mathcal{C}}
\def\cF{\mathcal{F}}
\def\cO{\mathcal{O}}
\def\cH{\mathcal{H}}
\def\cM{\mathcal{M}}

\newenvironment{formula}[1]{\begin{equation}\label{eq:#1}}
                       {\end{equation}\noindent}
\def\Fi#1{\begin{formula}{#1}}
\def\Ff{\end{formula}\noindent}

 \date{ }

\title{Weak dependence for a class of local functionals of Markov chains on $\Z^d$}

\author[1]{C. Boldrighini}
\affil[1]{\small{Dipartimento di Matematica G. Castelnuovo, Sapienza Universit\`{a} di Roma,
Piazzale Aldo Moro 5, 00185 Roma. \par GNFM, Istituto Nazionale di Alta Matematica, Piazzale Aldo Moro 5, 00185 Roma.}
\thanks{boldrigh@mat.uniroma1.it}}
\author[2]{A. Marchesiello}
\affil[2]{\small{Faculty of Nuclear Sciences and Physical Engineering,
Czech Technical University in Prague, D\v{e}\v{c}\'{\i}n Branch,
Pohranicn\'{\i} 1, 40501  D\v{e}\v{c}\'{\i}n}
  \thanks{anto.marchesiello@gmail.com}}
\author[3]{C. Saffirio}
\affil[3]{\small{Institut f\"{u}r Mathematik
Universit\"{a}t Z\"{u}rich,
Winterthurerstrasse 190, CH-8057 Z\"{u}rich
\thanks{ chiara.saffirio@math.uzh.ch}}}

\begin{document}

\maketitle
{\it Dedicated to the 90.th anniversary of Academician Yu. M. Berezansky}
\begin{abstract}
\noindent In many models of Mathematical Physics, based on the study of a Markov chain $\widehat \eta= \{\eta_{t}\}_{t=0}^{\infty}$ on $\Z^{d}$, one can prove by perturbative arguments a contraction property of the stochastic operator restricted to a subspace of local functions $\cH_{M}$ endowed with a suitable norm.
We show, on the example of a model of random walk in random environment with mutual interaction, that the condition is enough to prove a Central Limit Theorem for sequences  $\{f(S^{k}\widehat \eta)\}_{k=0}^{\infty}$, where $S$ is the time shift and  $f$ is  strictly local in space and belongs to a class of functionals related to the H\"older continuos functions on the torus $T^{1}$.

\end{abstract}

 \section {Introduction}
 \label {S1}

Many problems in Physics and other sciences lead to consider Markov chains on  the $d$-dimensional lattice $\Z^d$ with local interaction (see \cite{MaMi}).  The states of the chain are random fields   $\eta_t= \{\eta_t(x): x\in \Z^d\}$,  $ t\in \Z_+ = \{0, 1, \ldots\}$, with $\eta_t(x) \in S$, where $S$   is usually a finite or countable set. \par\smallskip
In many models, notably in the work of R.A. Minlos and collaborators (see, e.g. \cite{AMZ, BMP, BMP00, BMP2, KM, KMZh, MaMi} and references therein) one can prove, usually by perturbative arguments, the existence of an invariant measure $\Pi$ on the state space $\Omega = S^{\Z^d}$,  and of a subspace of local functions  $\cH_M\subset L_2(\Omega, \Pi)$, invariant with respect to  the  stochastic operator $\cT$  and  such that for all $F\in \cH_{M}$ with zero average   $\langle F \rangle_{\Pi} =0$,  we have, for some  constant  $\bar \mu\in (0,1)$,
 \be \label{uno} \left | (\cT F)(\xi)\right | \; \leq \;   \bar\mu \| F \|_{M}, \qquad \xi\in \Omega .\ee
Here  $\|\cdot \|_M$ is   a   suitable  norm  on $\cH_{M}$, which is as a rule identified with the help of an expansion in a natural basis.

    \par\smallskip

  If one considers sums of functionals depending on the  space-time field $\widehat \eta = \{\eta_{t}\}_{t=0}^{\infty} \in \widehat \Omega = S^{\Z^{d}\times \Z_{+}}$, $ \Z_{+} = \{0,1,2,\ldots\}$,  of the type $\sum_{t=0}^{T} f(S^{t}\widehat \eta)$, where $S$ is the time shift, $S\widehat \eta = \{\eta_{t}\}_{t=1}^{\infty}$ and $f$ is a functional which is local in space, one cannot in general obtain   a Central Limit Theorem (CLT) by relying on properties such as strong  mixing and the like \cite{Ibr71}, which need requirements that   may  not apply or may be difficult to prove  \cite{Davy}.   The aim of the present paper is to establish   properties which hold in the framework described above and are sufficient for the CLT to hold.  \par\smallskip
   The models to which our description above applies are of different nature, and  the space $\mathcal H_{M}$ is based  on explicit constructions, so that it is  convenient to  work on the example of a particular  model. The model that we consider here is  a random walk in dynamical random environment with mutual interaction introduced in the papers  \cite{BMP0, BMP}:   the Markov chain $\eta_t, t\in \Z_+$, describes the  ``environment from the point of view of the random walk'', an object which plays an important role in the analysis of random walks in random environment \cite{KV}.   \par\smallskip
   Our results are inspired by  a classical result on the CLT for functionals of independent variables by Ibragimov and Linnik \cite{Ibr71}  (Th. 19.3.1).\par\smallskip

  In the next section we describe the model, which is a perturbation of an independent model, and present the main features which are relevant to our analysis.
   In \S 3 we   prove some preliminary results and  in the final section  \S 4 we prove our main results.
  \par\smallskip

 \section {Description of the model}
 \label {S2}

We  consider a  version of the model studied in   \cite{BMP0,BMP},  which describes a discrete-time random walk $X_{t} \in \Z^d$, $d\geq1$,  $t\in \Z_+$, evolving in mutual  interaction with a random  field   $\xi_t=\{\xi_t(x)\,:\,x\in\Z^d\}$, with   $\xi_t(x)\in S = \,\{\pm 1\}$. The state space  is $\Omega = S^{\Z^{d}}$, and the space of the "trajectories" (or "histories") of the environment $\hat{\xi}=\{\xi_t \ :\ t\in\Z_+ \} $ is    $\widehat \Omega = S^{\Z^{d}\times \Z_{+}}$. Measurability is  understood  with respect to the $\sigma$-algebra generated by the cylinder sets.\par\smallskip
The pair $(X_{t}, \xi_{t}), t\in \Z_{+}$ is a conditionally independent Markov chain \cite{MaMi}, i.e.,  if  $A\subset \Omega$ is a measurable set,  we have
\begin{equation}\label{conditional} 
\begin{split}
P(X_{t+1}&=x+u, \xi_{t+1} \in A\ | \ X_t=x, \xi_t = \bar \xi) \\ 
&= P(X_{t+1}=x+u\ | \ X_t=x,  \xi_t = \bar \xi) \;  P( \xi_{t+1} \in A\ | \ X_t=x, \xi_t = \bar \xi)\,.
\end{split} 
 \end{equation}

 If    $\hat{\xi}\in \widehat \Omega $ is fixed,  the first factor on the right of \eqref{conditional} defines the "quenched" random walk, for which we assume    the   simple  form
\begin{equation}\label{probability-1}
P(X_{t+1}=x+u\ | \ X_t=x,  \xi_t = \bar \xi)=P_0(u)+ \epsilon  c(u) \bar \xi(x) , \;\;\; u\in\Z^d, \bar \xi \in \Omega.
\end{equation}
Here $\epsilon >0$ is a small parameter,   $P_0$ is a probability  distribution on $\Z^{d}$ and $c$  is a real function  on $\Z^{d}$, such that  $P_0(u) \pm \epsilon c(u) \in [0, 1)$,  $u\in \Z^d$.
 We also assume that $P_{0}$  is even  and $c$ odd in $u$,  and that both are finite range.  By homogeneity in space it is not restrictive to assume $X_0=0$.
 \par\smallskip
 For the random walk transition probability   $P_{0}$, with characteristic function $\tilde p_0(\lambda) = \sum_{u\in \Z^d} P_0(u) e^{i (\lambda, u)}$ we assume that it  is  non-degenerate, i.e.,  $|\tilde p_0(\lambda)|=1$ if and only if $\lambda=0$, and, in order to meet a technical assumption in \cite{BMP}, we also need that the Fourier coefficients of the function $1\over \tilde p_0(\lambda)$ are absolutely summable.      \par\smallskip
 The evolution of the environment  is independent at each site, so that  $P( \xi_{t+1} \in A\ | \ X_t=x, \xi_t = \bar \xi) $ is a sum of products of the factors
 \begin{equation}\label{conditionalenv}  P( \xi_{t+1}(y) = s | \ X_t=x, \xi_t = \bar \xi)  = (1-\delta_{x,y}) Q_0(\bar \xi(y), s) + \delta_{x,y} Q_1(\bar \xi(y), s) \end{equation}
 where $s=\pm 1$,  $Q_0, Q_{1}$ are symmetric $ 2 \times 2$ matrices, $Q_{0}$ has   eigenvalues $1, \mu$, $|\mu|\in (0,1)$, and $Q_1$ is such that  $Q_1-Q_0 =\cO(\epsilon)$.
 In words, at each site $x\in \Z^{d}$ the evolution is given by the  transition matrix $Q_0$,  except at the site  where the random walk is located,  where the transition matrix is $Q_1$.
  \par\smallskip

   A natural probability measure on the state space $\Omega $ is the product   $\Pi_0=\pi_0^{\Z^d}$, with $\pi_{0} = (1/2, 1/2)$. If $Q_0 = Q_1$ (no reaction on the environment) $\Pi_{0}$ is   invariant. \par \smallskip

  The model just described  was first considered in \cite{BMP} both for the annealed and quenched case.  If there is no reaction on environment (i.e., $Q_{0}= Q_{1}$)  the CLT  for the annealed and quenched asymptotics of the random walk was obtained in a  general setting  \cite{DKL}. A  non-perturbative result was obtained in \cite{DL}.

   \par\smallskip
The field $\eta_t(x)=\xi_t(X_t+x)$, $t\in \Z_{+}$ is the  ``environment from the point of view of the particle''. $\{\eta_t\,:\,t\in\Z_+\}$ is also a  Markov chain with state space $\Omega$, and it can be shown \cite{BMP2, DL} that it is equivalent to the full process $(X_t, \xi_t)$, i.e, for all $T \in \Z_{+}$, $T\geq 1$, given the  sequence $\eta_0,\dots,\eta_T$ one can reconstruct $(X_0,\xi_0),\dots,(X_T,\xi_T), $ almost-surely.

 \par\smallskip
 The stochastic operator $\cT$ on the Hilbert space $\mathcal H=L_2(\Omega; \Pi_0)$, is defined as
 \begin{equation} \label{eq:tr-matrix} (\mathcal T f)(\bar\eta)=\langle f(\eta_{t+1})|\eta_t=\bar\eta\rangle, \;\;\; f\in\mathcal H
\end{equation}
where the average $\langle\cdot\rangle$ is w.r.t. the transition probability \eqref{probability-1}.  By  our assumptions  $\mathcal T$ preserves parity under the exchange $\eta \to - \eta$.\par\smallskip
 In $\cH$ we introduce a convenient   basis.
As $Q_{0}$ is symmetric, its eigenvectors are   $e_0=(1,1)$ and $e_1=(1,-1)$ with corresponding eigenvalues $1$ and $\mu$. We denote their components as $e_{j}(s)$,  so  that $e_1(s)= s$,  $e_0(s)= 1$,  $  s= \pm 1$, and set \begin{equation}\label{besis}
\Phi_{\Gamma}(\bar \eta)=\prod_{x\in \Gamma} e_1(\bar \eta(x))=\prod_{x\in \Gamma} \bar \eta(x), \qquad \Gamma \in  \frak G, \end{equation}
where $\frak G$ is the collection of the finite subsets of $\Z^{d}$, with $\Phi_{\emptyset} =1$.
$\{\Phi_{\Gamma}: \Gamma \in \frak G\}$   is a discrete orthonormal complete basis in $\cH$, and for $f\in \cH$ we write    $f(\eta) =  \sum_{\Gamma\in \frak G} f_{\Gamma}\Phi_{\Gamma}$.
\par\smallskip
For $M>1$    the dense  subspace $\cH_{M}\subset \cH$ is defined as
 \be\label{eq:M-cond}
 \mathcal{H}_M=\{f=\sum_{\Gamma} f_{\Gamma}\Phi_{\Gamma}\ : \ \|f\|_M=\sum_{\Gamma}|f_{\Gamma}|M^{|\Gamma|}<\infty\}.
 \ee
 $\mathcal H_M$  equipped with the norm $\|\cdot\|_M$ is a Banach space. As  $|\Phi_{\Gamma}(\eta)| = 1$,  we have  \begin{equation}\label{normsup}\|f\|_{\mathcal{H}}\leq\| f\|_{\infty } \leq \|f\|_M, \qquad  f\in\mathcal H_M. \end{equation}
Moreover $\mathcal H_{M}$ is closed under multiplication. In fact, as it is to see,
 $$\Phi_{\Gamma} \Phi _{\Gamma^{\prime}} = \Phi_{\Gamma\triangle \Gamma^{\prime}}, \qquad  \Gamma\triangle \Gamma^{\prime} = \Gamma \setminus \Gamma^{\prime} \;\cup\; \Gamma^{\prime} \setminus \Gamma, $$
 so that if $f, g\in \mathcal H_{M}$ and  $f=\sum_{\Gamma} f_{\Gamma} \Phi_{\Gamma}$, $g=\sum_{\Gamma} g_{\Gamma} \Phi_{\Gamma}$, we have
 \begin{equation}\label{multiplication}  \|f g \|_{M} =  \sum_{\Gamma\Gamma^{\prime}} \left |  f_{\Gamma} g_{\Gamma^{\prime}} \right | M^{|\Gamma \triangle \Gamma^{\prime}|}  \leq \|f\|_{M} \|g\|_{M} .\end{equation}

 \par\smallskip

 In the paper \cite{BMP} an analysis of the  expression of the matrix elements of $\cT$ and its adjoint $\cT^*$,   relying on their spectral properties   for $\epsilon=0$, leads to
  the following results.\par
 \begin{theorem} \label{base} If $\epsilon$ and $|\mu|$ are small enough, the space $\cH_M$ is invariant under $\cT$, and there is an invariant probability measure $\Pi$ for the chain $\{\eta_t\}$ which is absolutely continuous with respect to  $\Pi_0$ with uniformly bounded density $v(\eta)$. Moreover $\cH_M$ can be decomposed as $$\mathcal{H}_M=\mathcal{H}^{(0)}_M+\widehat{\mathcal{H}}_M$$
where $\mathcal{H}^{(0)}_M$ is the space of the constants, and on  $\widehat{\mathcal H}_M$ the restriction of $\cT$ acts as a contraction:
\begin{equation}\label{mu-bar}
\|\mathcal{T} f\|_M\leq\bar{\mu} \|f\|_M, \ \ \ \qquad  f \in\widehat{\mathcal{H}}_M,
\end{equation}
 $ \bar{\mu}\in(0,1)$, $ \bar{\mu} = |\mu| + \cO(\epsilon)$. Furthermore if $f= f_0 + \widehat f, f_0\in \cH^{(0)}_M, \; \widehat f\in \widehat{\mathcal H}_M$,  then $$ f_0 = \int f(\eta) d\Pi(\eta)=\int f(\eta)v(\eta)d\Pi_0(\eta).$$
\end{theorem}

\section{Preliminary estimates}\label{Walsh}

We denote by $\mathcal {P}_{\Pi}$ the probability measure on $\widehat \Omega =\{\pm 1\}^{\Z^{d}\times\Z_{+}}$ generated by the initial distribution $\Pi$, and  by $\frak {M}_{t_0}^{t_1}$, $0\leq t_{0} \leq t_{1}$   the $\sigma$-algebra of subsets of $\widehat \Omega$ generated by $\{\eta_{t}\}_{t=t_{0}}^{t_{1}}$. As $\Pi$ is invariant, $\mathcal {P}_{\Pi}$ is  invariant under the time shift.\par\smallskip
We consider functionals $f$ which depend only on
the values of the field at the origin, i.e.,  on  the sequence of random variables $\{\eta_{t}(0)\}_{t=0}^{\infty}$.
  We set  for brevity   $\zeta_t=\eta_t(0)$  and  $\widehat \zeta = \{ \zeta_t: t\in \Z_+\}\in \Omega_{+} = \{\pm 1\}^{\Z_+}$.   $\mathcal {M}_{t_0}^{t_1}$, $0\leq t_0 < t_1$ will denote the $\sigma$-algebra  generated by the variables $\{\eta_t(0)\}_{t=t_0}^{t_1}$, which is a subalgebra of  $\frak {M}_{t_0}^{t_1}$.\par\smallskip

  By
    abuse of notation,   $f(\widehat \zeta)$  may denote a function on $\widehat \Omega$ or on $\Omega_{+}$, according to the circumstances, and similarly for the $\sigma$-algebras $\mathcal {M}_{t_0}^{t_1}$, $0\leq t_0 < t_1$. We also write $\cM_{t}$ and $\frak M_{t}$ for $\cM_{t}^{t}$ and $\frak M_{t}^{t}$, respectively.

  \par\smallskip
  In what follows if $f$ is a function on $\widehat \Omega$ we introduce the notation $\langle f(\cdot)\vert \frak M_{0}\rangle(\eta) = G^{(f)}(\eta)$, $\eta \in \Omega$. The following  lemma is a  simple  consequence of Theorem \ref{base}.
  \begin{lemma}\label{lemma1} Let $f(\widehat \zeta)$ be a cylinder function on $\Omega_{+}$, depending only on the variables $\zeta_{0},\ldots, \zeta_{m-1}$, $m\geq 1$. Then $ G^{(f)}(\eta)  \in \cH_{M}$ and
  \begin{equation} \label{agg} \left \|  G^{(f)} \right \|_{M} \leq \; C \; \max_{\gamma \in \{0,\ldots, m-1\}} |f_{\gamma}| (1+\mu_{*})^{m},
 \end{equation}
    where $\mu_{*} = M \sqrt {\bar \mu (1+ 2 \bar \mu)}$ and  $C>0$ is a  constant.
     \end{lemma}
    \begin{proof} \label{Proof.}  As $f$ depends only on $\zeta_{0},\ldots, \zeta_{m-1}$ it can be written in the form
 \begin{equation}\label{expansion} f(\widehat \zeta) = \sum_{\gamma\subset \{0, \ldots, m-1\}} f_{\gamma} \Psi_{\gamma}(\widehat \zeta) \end{equation}   where the sum runs over the subsets of $\{0, \ldots, m-1\}$, and the functions
 \begin{equation}\label{basis2} \Psi_\gamma (\widehat \zeta)  = \prod_{t\in \gamma} \zeta_t , \quad \gamma \neq \emptyset, \qquad \Psi_\emptyset(\widehat \zeta)  =1 \end{equation}
   are called ``Walsh functions''. The first assertion   follows from the fact that for any    subset  $\gamma = \{t_{0}, t_{1}, \ldots, t_{k}\} \subset \Z_{+}$, $t_{0}< t_{1}< \ldots  < t_{k}$,
  we have
  \begin{equation} \label{condprob} G_\gamma(\bar \eta) := \left \langle \Psi_\gamma | \frak M_{t_0} \right \rangle \in \cH_{M}, \qquad \bar \eta \in \Omega.  \end{equation}
  In fact, if    $r_{j}= t_{k-1-j}-t_{k-j}$, $j=1, \ldots, k $,    $G_\gamma$ can be written as
\begin{equation}\label{A2} G_{\gamma}(\bar \eta)  =  \Phi_{\{0\}}(\bar \eta) \left [ \cT^{r_{k}} \Phi_{\{0\}}\ldots \cT^{r_{1}} \Phi_{\{0\}}\right ](\bar \eta), \quad \bar \eta \in \Omega,  \end{equation}
i.e., $G_\gamma$ is obtained by successive applications of $\cT$ and of the multiplication operator by $\Phi_{\{0\}}$.  As both operations  leave $\cH_M$ invariant,  $G_{\gamma}\in \mathcal H_{M}$.
\par\smallskip

 Moreover the following inequality is proved in the Appendix
  \begin{equation} \label{lemmadue} \ \|G_\gamma\|_M \;  \leq  \; M^{|\gamma|}\; \bar \mu^{ [{|\gamma|\over 2}]} (1+ 2 \bar \mu)^{[{|\gamma|-1\over 2}]} \leq C \; \mu_{*}^{|\gamma|},\end{equation}  where     $[\cdot]$ denotes the integer part, and $C>0$ is a  constant which is easily worked out. \par\smallskip

The proof of the lemma follows by observing that the inequality \eqref{lemmadue} implies

 \begin{equation} \label{agg1} \left \| \langle f(\cdot)|\frak M_{0}\rangle \right \|_{M} \leq \; C \; \max_{\gamma \in \{0,\ldots, m-1\}} |f_{\gamma}| \sum_{\gamma\subset \{0, \ldots, m-1\}} \mu_{*}^{|\gamma|}.
 \end{equation}
\end{proof}

  We denote by $\wp$ the probability measure induced  by $\mathcal {P}_{\Pi}$ on $\Omega_{+}$. $\wp$ is   stationary with respect to the time shift on  $\Omega_+$:  $ S\widehat \zeta = \{\zeta_1, \zeta_2, \ldots \}$.\par\smallskip

  The following assertion is a simple consequence of the previous lemma.
\begin{lemma}\label{lemma1b} Under the assumptions of the previous lemma, if  $\bar \mu$ is so small that  $\mu_{*}   <1$,  then  the probability measure $\wp$  on $\Omega_+$  is continuous.
  \end{lemma}
\begin{proof} \label{Proof.} We need to prove that any point $\widehat \zeta^{(0)} =\{\bar \zeta_{k}\}_{k=0}^{\infty}\in \Omega_{+}$ has zero $\wp$-measure. Consider the cylinders  $Z_n(\widehat \zeta^{(0)}) = \{\zeta_{j} = \bar \zeta_{j}:\; j=0,1,\ldots, n-1\}$,  which are decreasing $Z_{n+1}(\widehat \zeta^{(0)})\subset Z_n(\widehat \zeta^{(0)})$ and   such that $\cap_{n} Z_{n}(\widehat \zeta^{(0)}) =\{\widehat \zeta^{(0)}\}$. The probabilities   \begin{equation}\label{probz}  \wp \left ( Z_n(\widehat \zeta^{(0)}) \right ) = {1\over 2^n} \left \langle \prod_{j=0}^{n-1} \left ( 1 + \bar \zeta_j \zeta_j \right ) \right \rangle_\wp \end{equation}
are computed by expanding the internal product   in terms of the functions $\Psi_\gamma$:
$$ \prod_{j=0}^{n-1} \left ( 1 + \bar \zeta_j \zeta_j \right )  =   \sum_{\gamma \subset \{0, \ldots, n-1\}} \Psi_{\gamma}(\widehat {\bar \zeta} ) \Psi_{\gamma}(\widehat \zeta), \qquad   \widehat {\bar \zeta}= \{\bar \zeta_{j}\}_{j=0}^{n-1}.$$
Recalling that $|\Psi_{\gamma}(\widehat \zeta)|=1$,   we have
$$\left | \left \langle  \sum_{\gamma \subset \{0, \ldots, n-1\}} \Psi_{\gamma}(\widehat {\bar \zeta} ) \Psi_{\gamma}(\widehat \zeta) \right \rangle_{\wp}\right |\leq  \sum_{\gamma \subset \{0, \ldots, n-1\}} \left |\left \langle  \Psi_{\gamma}(\widehat \zeta) \right \rangle_{\wp} \right | =$$

$$ =  \sum_{\gamma \subset \{0, \ldots, n-1\}}  \left |  \left \langle  \left \langle \Psi_{\gamma}|\frak M_{0}\right \rangle(\cdot) \right \rangle_{\Pi}\right | =  \sum_{\gamma \subset \{0, \ldots, n-1\}} \left |  \left \langle G_{\gamma}(\cdot)\right \rangle_{\Pi} \right | . $$
Therefore by the inequality \eqref{lemmadue}  the right side  is bounded by
$${C\over 2^{n}}  \sum_{\gamma \subset \{0, \ldots, n-1\}}  \mu_{*}^{|\gamma|} =  C \left ({1+\mu_{*}\over 2 }\right )^{n}.$$
Hence if $\mu_{*} <1$, the right side tends to $0$ as $n\to\infty$, which proves the lemma.
\end{proof}
From now on we assume that $\mu_{*}<1$.
\par\smallskip

  We pass to consider functions for which the expansion \eqref{expansion} is infinite, i.e., $\gamma$ runs over the collection
   $\frak g$  of the finite subsets $\Z_{+}$.  The functions $\{\Psi_{\gamma}: \gamma \in \frak g\}$,  are an orthonormal   basis in  $L_2(\Omega_+, \wp_0)$, where $\wp_0 = \pi^{\Z_+} $ is the probability measure on $\Omega_+$ corresponding to the random variables $\{\zeta_k\}_{{k=0}}^{\infty}$ being i.i.d. with distribution $\pi(\pm 1) ={1\over 2}$. The corresponding series is called ``Fourier-Walsh expansion'' \cite{Fine}. \par
 \smallskip
  A map  $\cF: \Omega_+\to T^1$, where $T^1 =[0, 1) \mod 1$  is the one-dimensional torus, is defined by associating to a point $\widehat \zeta\in \Omega_+$ the binary expansion $x =0,a_0 a_1\ldots\in [0, 1]$,  with $a_t = {1-\zeta_t\over 2} $, $t\in \Z_+$.
 $\cF$ is not invertible because the dyadic points of $T^{1}$ have two binary expansions, but it becomes invertible if we exclude the sequences such that $\zeta_{t}= -1$ for all $t$ large enough.  Such sequences are a countable set, which has zero $\wp_0$-measure, and also, by Lemma \ref{lemma1b},  zero $\wp$-measure.
 \par\smallskip
  Under the map $\cF$ the  basis functions  $\Psi_{\gamma}$ go into the functions $$\psi_{\gamma}(x) = \prod_{t\in \gamma} \phi_{t}(x), \qquad \gamma \in \frak g .$$
where $\phi_{t}(x)$  is the image of    $\zeta_{t}$,  $t\in \Z_{+}$, i.e.,    $$\phi_{0}(x)  =\left\{
\begin{array}{ll}
1, & 0\leq x < {1\over 2}\\
-1, & {1\over 2}\leq x < 1
\end{array}\right. $$
and   for $t\geq 0$,    $\phi_{t}(x) = \phi_{0}(2^{t} x)$, where $2^{t} x$ is understood $\mod 1$.  \par\smallskip

If $f\in L_2(\Omega_{+}, \wp_{0})$ then $\tilde f(x) = f (\cF^{-1} x) \in L^{2}(T^1, dx)$ and can be expanded in the orthonormal basis $\{\psi_\gamma: \gamma \in \frak g\}$, with coefficients
\begin{equation} \label{coeff} f_{\gamma} = \int_{\Omega_{+}} f(\widehat \zeta) d\wp_{0}(\widehat \zeta) = \int_{0}^{1} \tilde f(x) \psi_{\gamma}(x) dx .\end{equation}
\par\smallskip
  A natural way of ordering the collection $\frak g$ of the finite subsets of $\Z_{+}$, which plays an important role in the theory, is obtained by setting  $\gamma_0=\emptyset$ and $\gamma_n = \{ t_1, t_2, \ldots, t_r\}$, where  $r$ and $0 \leq t_1< t_2 < \ldots < t_r$ are uniquely defined by the relation $ n = 2^{t_1} + \ldots + 2^{t_r} $.
We  call Walsh series both the expansion
\begin{equation} \label{f1}  f(\widehat \zeta) =\sum_{\gamma \in \frak g} f_{\gamma} \Psi_{\gamma}(\widehat \zeta) =  \sum_{n=0}^\infty f_{\gamma_n} \Psi_{\gamma_n}(\widehat \zeta),  \end{equation}and the corresponding expansions of $\tilde f(x)$.
  For the latter, an important  role is played by a particular set of partial sums
\begin{equation} \label{ridotta}  \Sigma_{2^{k}}(\tilde f; x) = \sum_{\gamma \subset \{0, 1, \ldots, k-1\}} f_{\gamma} \psi_{\gamma}(x) = \sum_{n=0}^{2^{k}-1} f_{\gamma_{n}} \psi_{\gamma_{n}}(x)  \end{equation}
for which it can be seen \cite{Fine} that
\begin{equation}\label{ridotta}   \Sigma_{2^{k}}(\tilde f; x) = 2^{k} \int_{\alpha_{k}}^{\beta_{k}} \tilde f(y) dy, \qquad  \alpha_k = m\;  2^{-k}, \;  \beta_k = (m+1) 2^{-k} \end{equation}
where the integer $m$ is such that   $\alpha_{k} \leq x < \beta_{k}$.

\par\smallskip
   The following result is proved in \cite{Fine}. We repeat it here, with a shorter proof based on conditional probabilities.\par

\begin{lemma}\label{lemmaf} Let $\tilde f(x)$ be a bounded function. Then  its Walsh-Fourier coefficients   $f_{\gamma}$,  given by \eqref{coeff}, satisfy the following inequality
\begin{equation}\label{stimacoeff}  \left | f_{\gamma} \right |  \; \leq \;  {\omega(\tilde f; 2^{-n-1})\over 2^{n+2}}, \qquad n = max\{ t: t\in \gamma \}, \end{equation}
where     $\omega(f;\delta)$ is the modulus of continuity of $\tilde f$:
\begin{equation}\label{coeffcont} \omega(f;\delta) = \sup_{x, x\prime \in T^1\atop |x-x^{\prime}|=\delta} {|f(x) - f(x^{\prime})|\over \delta}. \end{equation} \end{lemma}
\par\noindent
 \begin{proof}[Proof.]  We have
$$ f_\gamma = \left \langle f(\widehat \zeta) \prod_{t\in \gamma} \zeta_t \right \rangle_{\wp_0}  = \left \langle \prod_{t\in \gamma\setminus \{n\}} \zeta_t  \; \; \left \langle f(\widehat \zeta) \zeta_n \vert \cM_{0}^{n-1}\right  \rangle \right \rangle_{\wp_0} .$$
Going back to $T^{1}$, and setting $x_{n} = {a_{0}\over 2}+ \ldots + {a_{n-1}\over 2^{n}}$, $a_{j}={1-\zeta_{j}\over 2}$, we have
$$\left |  \left \langle f(\widehat \zeta) \zeta_n \vert \cM_{0}^{n}\right  \rangle \right |  = 2^{n} \int_{x_{n}}^{x_{n}+ 2^{-n}} \tilde f(x) \left ( 1 - 2\phi_{n}(x) \right ) dx= $$
 \begin{equation}\label{decay} =  2^{n} \int_{x_{n}}^{x_{n}+ 2^{-n-1}} \left [ \tilde f(x) - \tilde f(x + 2^{-n-1}) \right ] dx , \end{equation}
from which,  taking into account \eqref{coeffcont}, the inequality \eqref{stimacoeff} follows immediately. \end{proof}
\par\smallskip

The results above allow us to prove the analogue of Lemma \ref{lemma1} for functions $f$ such that  $\tilde f (x) = f(\cF^{-1}x)$ is H\"older continuous: $\tilde f \in \cC^{\alpha}(T^{1})$,  $\alpha\in (0,1)$.  In what follows if $g\in \cC^\alpha(T^1)$ we denote by $\|g\|_{\cC^{\alpha}}$ the norm and by  $\|g \|_\alpha$ the semi--norm
\begin{equation*}
\|g\|_{\alpha}=\sup_{x,y\in T^1}\frac{|g(x)-g(y)|}{|x-y|^\alpha}.
\end{equation*}

\begin{lemma}\label{lemma:hoelder}
Let  $f$ be a function on $\Omega_+$,   such that   $\tilde f\in \mathcal C^\alpha(T^1)$,  $\alpha \in (0, 1)$. If $\bar \mu$ is so small that  $\kappa := 2^{-\alpha}(1+\mu_*) < 1$, then
  $G^{(f)} \in \cH_M$ and  the  following inequality holds
   \begin{equation}\label{norma}\left \| G^{(f)} \right  \|_M \leq {C_{\a}\over 1- \kappa}  \|\tilde f\|_{\cC^{\alpha}} ,  \end{equation}    where $C_{\a}>0$  is a positive constant.
\end{lemma} \par\smallskip
\begin{proof}\label {Proof.}  If $2^k \leq n < 2^{k+1}$ the Fourier coefficient $\gamma_{n}$ in the Walsh series \eqref{f1} is such that $\max \{t\in \gamma_n\} = k$. Hence, as   $\delta \; \omega(\tilde f; \delta) \leq \delta^\alpha \|\tilde f\|_\alpha $, the inequality \eqref{stimacoeff} gives
\begin{equation}\label{correz} \left | f_{\gamma_n} \right | \leq { \|\tilde f\|_\alpha\over 2^{1+\a}} 2^{- k \alpha}, \qquad  2^k \leq n < 2^{k+1}.\end{equation}
Therefore we have
$$ \left \| \sum_{n= 2^k}^{2^{k+1}-1} f_{\gamma_n} \left \langle \Psi_{\gamma_n}|\frak M_0\right \rangle \right \|_M \leq   \frac{||\tilde f||_{\a}}{2^{1+\a}}   2^{- k \alpha}   \sum_{n= 2^k}^{2^{k+1}-1}  \left \|  \left \langle \Psi_{\gamma_n}|\frak M_0\right \rangle \right \|_M .$$
Observe moreover that the number of elements of $\gamma_n$ is  $r_n = |\gamma_n| = u_n -1$ where $u_n$ is the number of "$1$" in the binary expansion of $n$. Hence, by the inequality \eqref{lemmadue} we find
$$ \sum_{n= 2^k}^{2^{k+1}-1}  \left \|  \left \langle \Psi_{\gamma_n}|\frak M_0\right \rangle \right \|_M \leq C \sum_{s=0}^{k-1} { k-1 \choose s} \mu_*^s = C (1+ \mu_*)^{k-1} ,$$
which, as $|f_{\emptyset}|\leq \|f\|_{\infty}$, together with \eqref{correz},  implies \eqref{norma}.   \end{proof}

  \par\smallskip

  \section{Weak dependence and the Central Limit Theorem}\label{WDCLT}
In the present paragraph we prove our main results for sums of sequences  of the type
$f(S^{t}\widehat \zeta)$, $t=0, 1,\ldots$. As $\mathcal P_{\Pi}$ and the measure $\wp$ induced by it on $\Omega_{+}$ are invariant under time shift, the sequence is stationary in distribution.\par\smallskip

In what follows   $\langle \cdot \rangle$ denotes an average with respect to $\wp$, $\mathcal P_{\Pi}$ or $\Pi$, according to the context.  Moreover we denote by $c_{i}, i=1,2,\ldots$, and sometimes by $\rm const$,   different constants which depend on the parameters of the model.
\par\smallskip
Let $f$ be a bounded measurable function on $\Omega_{+}$ with $\langle f\rangle_{\wp}=0$, and
\begin{equation}\label {somma0} S_{n}(\widehat\zeta\vert f) =   \sum_{t=0}^{n-1} f(S^{t} \widehat \zeta), \qquad n=1,2,\ldots . \end{equation}
If $f$ admits a Walsh expansion \eqref{f1} then    $ \sum_{\gamma \in \frak g} f_{\gamma}\langle \Psi_{\gamma} \rangle_{\wp} = \langle f \rangle_{\wp}=0$, so that
\begin{equation}\label{somma1} f(\widehat \zeta) = \sum_{\gamma\in \frak g,\; \gamma \neq \emptyset} f_{\gamma} \widehat \Psi_{\gamma}(\widehat \zeta), \qquad \widehat \Psi_{\gamma}(\widehat \zeta) = \Psi_{\gamma}(\widehat \zeta) - \left \langle \Psi_{\gamma}( \cdot)\right \rangle_{\wp}. \end{equation}
\par\smallskip

In what follows we make repeated use of the fact that if  $f$ is a function on $\widehat \Omega$ and $G^{(f)}:= \langle f(\cdot)\vert \frak M_{0}\rangle\in \cH_{M}$,  then, by   Theorem \ref{base}, $ \langle f(S^{t+h}\cdot)\vert \frak M_{h}\rangle = \cT^{t}G^{(f)}\in \cH_{M}$. \par\smallskip

\begin{theorem}\label{theorem:clt}  Let $f$ be a  function on $\Omega_{+}$, depending only on $\zeta_{0}, \ldots,\zeta_{m-1}$,  $m\geq 1$,  and such that $\langle f \rangle_{\wp} =0$. Then  the dispersion of normalized  sums
$ S_{n}(\widehat \zeta\vert f)\over \sqrt n$
 tends, as $n\to \infty$, to a finite non-negative  limit   \begin{equation}\label {dispersione} \sigma^{2}_{f} = \left \langle f^{2}(\cdot) \right \rangle_{\wp} + 2 \sum_{t=1}^{\infty} \left \langle f(\cdot) f(S^{t}\cdot)  \right \rangle_{\wp}    \end{equation}
 and the series is absolutely convergent. Moreover,  if $\sigma^{2}_{f}>0$,   the sequence $S_{n}(\widehat\zeta\vert f)\over \sqrt n$ tends weakly  to the centered gaussian distribution with  dispersion $\sigma^{2}_{f}$.
  \end{theorem}
\begin{proof}\label {Proof.}

The proof of the theorem  is based on two basic inequalities.
 \begin{equation}\label{ineq1} \left \|  \left \langle f(\cdot) f(S^{t}\cdot)\vert \frak M_{0} \right \rangle \right \|_{M}   \leq c_{1}\|f\|^{2}_{\infty} \bar \mu^{\max \{0, t-m+1\}}  (1+ \mu_{*})^{2m},\end{equation}
 \begin{equation} \label{giesse} \left \| G^{(S_{n})}\right \|_{M } \leq  c_{2}\;  \|f\|_{\infty} \;  (1+ \mu_{*})^{m}, \qquad  \left \| G^{(\widehat S^{2}_{n})}\right \|_{M } \leq  c_{3}\;  \|f\|^{2}_{\infty} \;m  (1+ \mu_{*})^{m}, \end{equation}
 where $\widehat S^{2}_{n}(\widehat \zeta\vert f) = S^{2}_{n}(\widehat \zeta\vert f) - \langle S^{2}_{n}(\cdot\vert f) \rangle $ and  $c_{1}, c_{2}, c_{3}$ are  constants independent of $m$. \par\smallskip
 For the proof of \eqref{ineq1}, observe that  if $t\leq m-1$, then $ f^{(2)}_{t}(\widehat \zeta):=f(\widehat \zeta) f(S^{t}(\widehat \zeta))$ is a cylinder function $\cM_{0}^{t+m-1}$- measurable and
     bounded by $\|f\|^{2}_{\infty}$. Hence, by Lemma \ref{lemma1} and \eqref{coeff}, $\|f^{(2)}_{t}\|_{M}\leq \cst \|f\|^{2}_{\infty}(1+\mu_{*})^{m+t}$, and \eqref{ineq1} holds for $t\leq m-1$.\par
     If $t\geq m$  taking the expectation with respect to $\frak M_{0}^{m-1}$ we have
  $$  \left \langle f(\cdot) f(S^{t}\cdot)\vert \frak M_{0} \right \rangle =  \left \langle f(\widehat \zeta) [\cT^{t-m+1} G^{(f)}](\eta_{m-1})\vert \frak M_{0} \right \rangle .$$
As  $ G^{(f)} \in \widehat \cH_{M}$, expanding $f$ in Walsh series and using Proposition \ref{proposition} in the Appendix and Lemma \ref{lemma1} we see that Inequality \eqref{ineq1} also holds for  $t\geq m-1$.

\par\smallskip
The first inequality in  \eqref{giesse} is a simple consequence of the  inequality $\|\langle f(S^{t}\cdot)\vert \frak M_{0}\|_{M} =\| \cT^{t}G^{(f)}\|_{M}\leq \bar \mu^{t} \|G^{(f)}\|_{M}$, of Lemma \ref{lemma1} and of the inequality $|f_{\gamma}| \leq \|f\|_{\infty}$ (see \eqref{coeff}). Moreover, setting $\widehat f^{2}(\widehat \zeta) = f^{2}(\widehat \zeta) - \langle f^{2}(\cdot) \rangle $  and $\widehat f^{(2)}_{t}(\widehat \zeta) =  f^{(2)}_{t}(\widehat \zeta) - \langle  f^{(2)}_{t}(\cdot) \rangle$, we have
 \begin{equation}\label{ineq2} \left  \|G^{(\widehat S^{2}_{n})}\right \|_{M} \leq \sum_{j=0}^{n-1} \left \| \cT^{j} G^{(\widehat f^{2})}\right \|_{M} +  2 \sum_{j=0}^{n-2} \sum_{t=1}^{n-j-1}      \left \| \cT^{j} G^{(\widehat f_{t}^{2})}\right \|_{M}, \end{equation}
 and the second inequality in \eqref{giesse} follows by observing that $\widehat f^{2}$ is a cylinder function with zero average and $\|\widehat f^{2}\|_{\infty}\leq \|f\|_{\infty}^{2}$, and using the estimate \eqref{ineq1} for $\|G^{(\widehat f_{t}^{2})} \|_{M}$.

 \par\smallskip
Passing to the assertions of the theorem, observe first that, by  the property \eqref{normsup} of $\cH_{M}$,  the absolute convergence  of the series  in \eqref{dispersione}  follows
   from Inequality \eqref{ineq1}.
  \par\smallskip
Assuming that $\sigma_{f}^{2}>0$, for  the proof of the CLT we adopt a Bernstein scheme. Let $p_{n} = [n^{\beta}]$, $q_{n} = [n^{\delta}]$, $k_{n} =[{n\over p(n) + q(n)}]$, with $0< \delta < \beta  < 1/4$.
The interval of integers $[0, n-1]$ is divided into   subintervals of length $p_{n}$ and $q_{n}$: $$I_{\ell} = [(\ell-1)(p_{n}+q_{n}), \; \ell p_{n} +(\ell-1) q_{n}-1] ,\quad J_{\ell}=[ \ell p_{n}+ (\ell-1) q_{n}, \ell (p_{n}+q_{n})-1],$$ $\ell = 1, \ldots k_{n}$, and the rest $J_{*}= [0, n-1]\setminus \cup_{j=1}^{k_{n}} [I_{j}\cup J_{j}]$. \par\smallskip

The sum \eqref{somma0}  is then written as $S_{n}(\widehat\zeta\vert f) = S^{(M)}_{n}(\widehat\zeta\vert f) + S^{(R)}_{n}(\widehat\zeta\vert f)$ where
\begin{equation}\label{splits}  S^{(M)}_{n}(\widehat\zeta\vert f) = \sum_{\ell=1}^{k_{n}} S_{I_{\ell}}(\widehat\zeta\vert f), \qquad S^{(R)}_{n}(\widehat\zeta\vert f) =
 \sum_{\ell=1}^{k_{n}} S_{J_{\ell}}(\widehat\zeta\vert f) + S_{J_{*}}(\widehat\zeta\vert f) \end{equation}
and $S_{I_{\ell}}, S_{J_{\ell}}, S_{J_{*}}$ denote the sums over the corresponding subinterval. \par\smallskip

We first prove that  the $L_{2}$-norm of  $ S^{(R)}_{n}/\sqrt n$ vanishes as $n\to \infty$, i.e.,
\begin{equation}\label{estim1}  \left \langle \left (  \sum_{\ell=1}^{k_{n}} S_{J_{\ell}}(\cdot\vert f) \right )^{2} \right \rangle  = k_{n} \left \langle S^{2}_{J_{1}}(\cdot\vert f) \right\rangle + 2 \sum_{1\leq s < t \leq k_{n}} \left \langle S_{J_{s}}(\cdot|f) S_{J_{t}}(\cdot|f) \right \rangle = \cO (n^{1-\beta +\delta}).\ee

For the proof, observe that  Inequality \eqref{ineq1} implies that $\langle S^{2}_{J_{1}}(\cdot|f) \rangle = \cO(q_{n})$, so that the first term on the right of \eqref{estim1} is of the order $k_{n} q_{n} \sim n^{1-\beta +\delta}$. \par
For the second term, observe   that, by translation invariance,  recalling that  $S_{q_{n}}(\widehat \zeta\vert f)$ is $\frak M_{0}^{q_{n}+m-2}$-measurable and taking the corresponding conditional probability,  \begin{equation}\label{elle}\langle S_{J_{s}}(\cdot|f) S_{J_{t}}(\cdot|f)\vert \frak M_{0}  \rangle =    \left \langle S_{q_{n}}(\cdot|f) \left [ \cT^{(t-s)\ell_{n} + p_{n}-m +2} \; G^{(S_{q_{n}})}\right ](\eta_{q_{n}-m +2})\vert \frak M_{0} \right \rangle ,\end{equation}
   where $\ell_{n}= p_{n}+q_{n}$.
   Therefore,  recalling the inequalities \eqref{normsup}, we get the estimate
\begin{equation}\label{ineq3} \left | \left \langle S_{J_{s}}(\cdot|f) S_{J_{t}}(\cdot|f) \right \rangle \right | \leq  {\rm const} \; (1+\mu_{*})^{m} \; \|f\|^{2}_{\infty} \; q_{n} \bar \mu^{(t-s) \ell_{n}+ p_{n}- m}.\end{equation}
 As  $k_{n} q_{n} \bar \mu^{p_{n}} \leq {\rm const} \; \bar \mu ^{p_{n}\over 2}$,  the double sum on the right of  \eqref{estim1} is of the order $\cO(\bar \mu^{p_{n}\over 2})$, so that \eqref{estim1} is proved.\par \smallskip
As for $S_{J_{*}}$,  \eqref{ineq1}  implies $\langle S^{2}_{J_{*}}(\cdot|f)\rangle\leq  \langle S^{2}_{p_{n}+q_{n}}(\cdot|f)\rangle= \cO(n^{-1+\beta})$.   This fact, together with \eqref{estim1}, proves that  $\langle (S^{(R)}_{n}(\widehat\zeta\vert f))^{2}\rangle /n = \cO(n^{-\beta+\delta})$, and, as  $\beta >\delta$,    $S^{(R)}_{n}$ does not contribute to the limiting distribution.

\par\smallskip
We now show that the random variables $\{S_{I_{\ell}}\}_{\ell=1}^{k_{n}}$ are  almost independent for large $n$, i.e., for the characteristic functions $\phi^{(\ell)}_{n}(\lambda|\widehat \zeta)= \exp\{i{\lambda\over \sqrt n}S_{I_{\ell}}(\widehat \zeta|f)\}$ we have
    \begin{equation}\label{indep} \left \langle \prod_{\ell=1}^{k_{n}} \phi^{(\ell)}_{n}(\lambda|\widehat \zeta) \right \rangle - \prod_{\ell=1}^{k_{n}} \left \langle  \phi^{(\ell)}_{n}(\lambda|\widehat \zeta)
 \right \rangle \to 0, \qquad n\to \infty.\end{equation}
 We proceed by  iteration. As a first step we consider the difference
$$  \left \langle \prod_{\ell=1}^{k_{n}} \phi^{(\ell)}_{n}(\lambda|\widehat \zeta) \right \rangle -  \left \langle \prod_{\ell=1}^{k_{n}-1} \phi^{(\ell)}_{n}(\lambda|\widehat \zeta) \right \rangle \left \langle \phi^{(k_{n})}_{n}(\lambda|\widehat \zeta) \right \rangle  = $$
\begin{equation}\label{quasind}  =\left \langle \prod_{\ell=1}^{k_{n}-1} \phi^{(\ell)}_{n}(\lambda|\widehat \zeta) \;  \widehat \phi^{(k_{n})}_{n}(\lambda|\widehat \zeta)    \right \rangle ,  \qquad \widehat  \phi^{(\ell)}_{n}(\lambda|\widehat \zeta) =  \phi^{(\ell)}_{n}(\lambda|\widehat \zeta) - \left \langle  \phi^{(\ell)}_{n}(\lambda|\widehat \zeta) \right \rangle. \end{equation}

We expand $\widehat \phi^{(k_{n})}_{n}(\lambda|\widehat \zeta)$ in Taylor series at $\lambda=0$, we have, for some $\lambda_{*}$, $|\lambda_{*}| \leq |\lambda|$,
\begin{equation}\label{3exp}  \widehat  \phi^{(k_{n})}_{n}(\lambda|\widehat \zeta)= i{\lambda\over \sqrt n} S_{I_{k_{n}}}(\widehat \zeta|f)  - {\lambda^{2}\over 2n} \left (S^{2}_{I_{k_{n}}}(\widehat \zeta|f) -\left \langle (S^{2}_{I_{k_{n}}}(\cdot|f) \right \rangle \right ) + i^{3}{\lambda^{3}\over n^{3\over 2} 3!} R_{n}(\lambda_{*}, \widehat \zeta),\end{equation}
\begin{equation} \label{expanding}  R_{n}(\lambda_{*}, \widehat \zeta) =     S^{3}_{I_{\ell}}(\widehat\zeta\vert f) \exp \{i {\lambda_{*}\over \sqrt n} S_{I_{\ell}}(\widehat\zeta\vert f)\} - \left \langle S^{3}_{I_{\ell}}(\widehat\zeta\vert f) \exp \{i {\lambda_{*}\over \sqrt n} S_{I_{\ell}}(\widehat\zeta\vert f)\}  \right \rangle  , \end{equation}
Clearly $|R_{n}(\lambda_{*}, \widehat \zeta)|\leq 2
 p^{3}_{n} |\lambda|^{3}\|f\|^{3}_{\infty} = \cO(n^{3 \beta})$, so that,
  as $\beta < 1/4$,  we need only consider the first two terms of the expansion \eqref{3exp}.\par\smallskip

 The product of the first $k_{n}-1$ factors in the expectation in \eqref{quasind} is measurable with respect to $\frak M^{t_{n}}_{0}$, where $t_{n} = (k_{n}-1) p_{n} + (k_{n}-2) q_{n} +m-2$. Taking the corresponding conditional expectation, by Inequality \eqref{giesse} we get for
 the first order term  the estimate
 \begin{equation}\label{ineq4}  \left | \left \langle \prod_{\ell=1}^{k_{n}-1} \phi^{(\ell)}_{n}(\lambda|\widehat \zeta) \left [\cT^{q_{n}-m+2}  G^{(S_{p_{n}})}   \right ](\eta_{t_{n}})    \right \rangle \right | \leq c_{4}\; \bar \mu^{q_{n}-m}\;  (1+\mu_{*})^{m} \|f\|_{\infty}. \end{equation}

 \par\smallskip

For the second order term, proceeding in the same way, and taking into account the second inequality \eqref{ineq2}  we come to the estimate
\begin{equation}\label{ineq4bis}  \left | \left \langle \prod_{\ell=1}^{k_{n}-1} \phi^{(\ell)}_{n}(\lambda|\widehat \zeta) \left [\cT^{q_{n}-m+2}  G^{(\widehat S^{2}_{p_n})}   \right ](\eta_{t_{n}})    \right \rangle \right | \leq c_{5}\; \bar \mu^{q_{n}-m}\;  m (1+\mu_{*})^{m} \|f\|_{\infty}^{2} . \end{equation}

 \par\smallskip
 Iterating the procedure for the remaining product $ \langle \prod_{\ell=1}^{k_{n}-1} \phi^{(\ell)}_{n}(\lambda|\widehat \zeta)  \rangle$, we see that the quantity on the left of \eqref{indep} is of the order $\cO(k_{n} n^{-3({1\over 2}-\beta)}) = \cO(n^{-{1\over 2} + 2\beta})$, so that , as $\beta < 1/4$, it vanishes as $n\to \infty$.
 \par\smallskip

 We are left with a sum $\tilde S^{(M)}_{n}$ of $k(n)$ independent variables distributed as $S_{p_{n}}(\widehat \zeta\vert f)$. The $\log$ of the characteristic function of the corresponding normalized sum is
\begin{equation}\label{charact} k_{n}\;  \psi_{n}({\lambda\over \sqrt n} \vert f),  \qquad  \quad \psi_{n}(\lambda\vert f) = \log \left \langle e^{i {\lambda } S_{q_{n}}^{(f)}(\cdot)}\right  \rangle . \end{equation}
Expanding  $\psi_{n}$ in Taylor series at $\lambda=0$, we see, in analogy to the proof above, that the third order remainder is of order $\cO(n^{-3({1\over2}-\beta)})$, so that it does not contribute to the limit. The first order term vanishes, and we see that
$$\lim_{n\to\infty} k_{n}  \psi_{n}({\lambda\over \sqrt n} \vert f) =-  {\lambda^{2}\over 2 } \;  \lim_{n\to \infty} {k_{n}\over n} \left [ p_{n} \left \langle f^{2}(\cdot) \right \rangle + 2 \sum_{j=0}^{p_{n}-1} \sum_{k=1}^{p_{n}-j-1} \left \langle f(\cdot) f(S^{k})\cdot \right \rangle \right ] .$$
As   ${k_{n}p_{n}\over n} \to 1$, the expression on the right tends to $-{\lambda^{2}\over 2} \sigma^{2}_{f}.$
  The theorem is proved.
 \end{proof}

 \begin{theorem}\label{theorem:cltalpha}  Let $f$ be a  function on $\Omega_{+}$, satisfying the assumptions of Lemma \ref{lemma:hoelder} with $\alpha > 1/2$ and such that $\langle f \rangle_{\wp} =  0$.  Then, if $\bar \mu$ is  small enough,  the dispersion of the normalized  sums $ S_{n}(\widehat \zeta\vert f)\over \sqrt n$
 tends, as $n\to \infty$, to a finite non-negative  limit   \begin{equation}\label {dispersione1} \sigma^{2}_{f} = \left \langle f^{2}(\cdot) \right \rangle_{\wp} + 2 \sum_{t=1}^{\infty} \left \langle f(\cdot) f(S^{t}\cdot)  \right \rangle_{\wp}    \end{equation}
 were the series on the right is absolutely convergent. Moreover,  if $\sigma^{2}_{f}>0$,   the sequence $S^{(f)}_{n}(\widehat \zeta)$ tends weakly  to the centered gaussian distribution with  dispersion $\sigma^{2}_{f}$.
  \end{theorem}
  \begin{proof} \label {Proof}   The proof repeats the pattern of the previous proof, to which we refer.  Inequalities \eqref{ineq1} and   \eqref{giesse} are replaced by
 \begin{equation} \label{nnjeravv} \left \| \left \langle f(\cdot) f(S^{t}\cdot) |\frak M_{0}\right \rangle\right |_{M} \leq c_{6}  \|\tilde f\|_{\cC^{\alpha}}^{2} \kappa^{t}, \end{equation}
 \begin{equation} \label{giesseb} \left \| G^{(S_{n})}\right \|_{M } \leq  c_{7}\;  \|\tilde f\|_{\cC^{\alpha}},    \qquad  \left \| G^{(\widehat S^{2}_{n})}\right \|_{M } \leq  c_{8}\;  \|\tilde f\|_{\cC_{\alpha}}^{2}.\end{equation}
 \par\smallskip
 The proof of the estimate \eqref{nnjeravv} is  deferred to the Appendix. The first inequality \eqref{giesseb} is proved as in the previous theorem, recalling Lemma \ref{lemma:hoelder}. \par\smallskip
 The second inequality  \eqref{giesseb} follows from Inequality \eqref{ineq2}, observing that $\widehat f^{2}(\cF^{-1}\cdot) \in \cC^{\alpha}$ and using   Inequality \eqref{nnjeravv}.
 \par\smallskip
 For  the estimate \eqref{estim1} observe that \eqref{nnjeravv} again implies that  $\langle S^{2}_{J_{1}}(\cdot|f) \rangle = \cO(q_{n})$. For the second term on the right of \eqref{estim1}
   we need, as in  \cite{Ibr71},  that the functions are well approximated by their conditional probabilities on finite $\sigma$-algebras.
This property is provided by the representation \eqref{ridotta} for the partial sums,  which gives
  \begin{equation} \label{njerav5} \left | f(\widehat \zeta) -  \Sigma_{2^{n}}(f; \widehat \zeta ) \right | \leq \|\tilde f \|_{\alpha} 2^{- \alpha n}. \end{equation}
 Let $m_{n} =[ {4\over \alpha} \log_{2} n]$, where $[\cdot]$ denotes the integer part. In the expression  \eqref{elle}, in the sum $S_{J_{s}}(\widehat \zeta\vert f)$,
 we replace the function $f$ by its partial sum $\Sigma_{2^{m_{n}}}$. The corresponding sum is denoted $\tilde S_{J_{s}}$.
 By Inequality \eqref{njerav5} we have
$$\left  \langle S_{J_{s}}(\cdot|f) S_{J_{t}}(\cdot|f)\vert \frak M_{0} \right \rangle  = \left  \langle \tilde S_{J_{s}}(\cdot|f) S_{J_{t}}(\cdot|f)\vert \frak M_{0} \right \rangle  + \cO\left ({q_{n}^{2} / n^{4}}\right ).$$
  $\tilde S_{J_{s}}(\cdot|f)$ can be treated as $ S_{J_{s}}(\cdot|f)$ in the previous proof,  so that the corresponding conditional expectation is written, if $n$ is so large that $p_{n}> m_{n}$, as
   \begin{equation}\label{elleb}    \left \langle \tilde S_{q_{n}}(\cdot|f) \left [ \cT^{(t-s)\ell_{n} + p_{n}-m_{n} +2} \; G^{(\tilde S_{q_{n}})}\right ](\eta_{q_{n}-m_{n} +2})\vert \frak M_{0} \right \rangle\end{equation} (the tilder in $\tilde S_{q_{n}}$ again denotes that $f$ is replaced by  $\Sigma_{2^{m_{n}}}$). By the first inequality \eqref{giesseb}
   $$\left | \left  \langle \tilde S_{J_{s}}(\cdot|f) S_{J_{t}}(\cdot|f)\vert \frak M_{0} \right \rangle \right | \leq \cst \; \|\tilde f	\|_{\cC^{\alpha}}^{2} \bar \mu^{(t-s)\ell_{n}+ p_{n}-m_{n}},$$
   and, as $k_{n}q_{n} \bar \mu^{p_{n}-m_{n}}\leq \cst \; \bar \mu^{p_{n}\over 2}$, we see that the same \eqref{estim1} holds in this case. The estimate for  $S_{I_{*}}$ is obvious, so that the negligibility of $S^{(R)}_{n}$ is proved.
   \par \smallskip

Further, we pass to the variables $\tilde S_{I_{\ell}}$, $\ell =1,\ldots, k_{n}$, obtained, as before, by replacing $f$ with the partial sum $\Sigma_{2^{m_{n}}}$. The correction is of order $\cO(n^{-3})$, so that it can be neglected. The rest of the proof repeats the previous steps, with the only changes that $m$ is replaced by $m_{n}$ and we use the estimates \eqref{giesseb}.
We omit the obvious details.
  \end{proof}

\section{Appendix}
\par\noindent
 {\bf Proof of inequality \eqref{lemmadue}.} Observe that, by symmetry  with respect to  the change of sign $\bar \eta(x)\rightarrow-\bar \eta(x)$, $x\in\Z^d$, the density $v(\bar \eta)$ is even. Moreover any finite trajectory of the Markov chain has the same probability of the trajectory obtained by  sign exchange.\par\smallskip
  The functions $\Phi_{\Gamma}$  defined by \eqref{besis}  are even (odd) for $|\Gamma|$ even (odd). Therefore  for $|\Gamma|$ odd we have $\langle \Phi_{\Gamma}\rangle_{\Pi}=0$, and  also $\langle \mathcal T^r\Phi_{\Gamma}\rangle_{\Pi}=0$, $r>0$. The functions $\Psi_{\gamma}$ are also even (odd) for  $|\gamma|$ even (odd), and for $|\gamma|$ odd $\langle \Psi_{\gamma}\rangle_{\wp} =\langle G_{\gamma}\rangle_{\Pi}=0$.
\par\smallskip
For $|\gamma|$ even we set
 \begin{equation}\label{decompo} G_{\gamma} = \langle G_{\gamma} \rangle_{\Pi} + \widehat G_{\gamma}, \qquad \widehat G_{\gamma}\in \widehat \cH_{M}. \ee
 If $\gamma = \{t_{0}, \ldots, t_{k}\} $, $k\geq 1$, we have, by \eqref{A2},
 $ G_{\gamma}(\bar \eta) = \Phi_{\{0\}}(\bar \eta)   [ \cT^{r_{k}} G_{\gamma\setminus \{t_{0}\}}  ](\bar \eta)$.
Therefore, if $|\gamma|\geq 2$ is even we have
\begin{equation}\label{even} \| G_{\gamma}\|_{M} \leq M \bar \mu^{r_{k}} \|G_{\gamma\setminus \{t_{0}\}}\|_{M}, \end{equation}
and if $|\gamma| > 1$ is odd
\begin{equation}\label{odd} \| G_{\gamma}\|_{M} \leq M \left ( | \langle G_{\gamma\setminus \{t_{0}\}}\rangle | +  \bar \mu^{r_{k}} \|\widehat G_{\gamma\setminus \{t_{0}\}}\|_{M} \right) \leq M (1 + 2\bar \mu^{r_{k}}) \| G_{\gamma\setminus \{t_{0}\}}\|_{M}, \end{equation}
where in the second inequality we take into account that $| \langle  G_{\gamma} \rangle | \leq \| G_{\gamma}\|_{\infty} \leq \|G_{\gamma}\|_{M}$.
 \par
\smallskip
   For    $|\gamma|= 1$,    $G_{\{t_{0}\}}(\bar \eta) = \Phi_{\{0\}}(\bar \eta)$ so that  $\|G_{\{t_{0}\}}\|_{M} = M$, and  for $|\gamma|=2$ we have
   $ \|G_\gamma \|_M \leq M \| \cT^{r_1} \Phi_{\{0\}} \|_M \leq M^2 \bar \mu^{r_1}$.
   Inequalities \eqref{even} and \eqref{odd} imply that
 \begin{equation}\label{iterazione} \|G_\gamma\|_M \leq M^{|\gamma|} \prod_{j\;  \rm{odd}}\bar \mu^{r_{j}}  \prod_{j\;  \rm{even}} (2\bar \mu^{r_{j}} +1) \end{equation}
which implies  \eqref{lemmadue}. \hfill \square
\par\smallskip

The following proposition is a simple consequence of the previous proof. \par
\begin{proposition} \label{proposition}  Under the assumptions of Lemma \ref{lemma1}, if $\gamma = \{t_0, \ldots, t_k\}$,  $G\in \widehat \cH_{M}$,  and  $t \geq t_{k}$, the following inequality holds, for some positive constant $C_{*}$.
 \begin {equation} \label{oddeven} \left \| \left \langle \Psi_{\gamma}(\zeta)\;  G(\eta_{t}) \vert  \frak M_{0}\right \rangle  \right \|_{M} \leq C_{*}\; \| G\|_{M} \; \bar \mu^{t-t_{k}} \; \mu_{*}^{|\gamma|}. \end{equation}
\end{proposition}
 \begin{proof} \label {Proof.}  Proceeding as in the previous proof,  we see that if $G$ is odd and $|\gamma|>1$, we get, in analogy with \eqref{iterazione},
 \begin {equation} \label{odd1}  \left \| \left \langle \Psi_{\gamma}(\zeta)\;  G(\eta_{t}) \vert  \frak M_{0}\right \rangle  \right \|_{M} \leq  M^{|\gamma|} \| G\|_{M}\; \bar \mu^{t-t_{k}} \prod_{j \;even} \bar \mu^{r_{j}} \prod_{j\; odd} (1+ 2\bar \mu^{r_{j}}),  \end{equation}
 and if $G$ is even, by  an obvious   modification of the   proof,
  \begin {equation} \label{even1}  \left \| \left \langle \Psi_{\gamma}(\zeta)\;  G(\eta_{t}) \vert  \frak M_{0}\right \rangle  \right \|_{M} \leq  M^{|\gamma|} \| G\|_{M}\; \bar \mu^{t-t_{k}} \prod_{j \;odd} \bar \mu^{r_{j}} \prod_{j\; even} (1+ 2\bar \mu^{r_{j}}). \end{equation}
 Writing  $G(\eta) = G^{(+)}(\eta) + G^{(-)}(\eta)$, where $G^{(\pm)}(\eta) = {G(\eta) \pm G(-\eta)\over 2} \in \widehat \cH_{M}$, and observing that
 $\|G\|_{M} =  \|G^{(+)}\|_{M} + \|G^{(-)}\|_{M}$, we get the result \eqref{oddeven}.
 \end{proof}. \par \smallskip
 \noindent {\bf Proof of Inequality \ref{nnjeravv}.}
   We denote by $m_{\gamma}, M_{\gamma}$ the minimum and maximum of the set $\gamma\in \frak g$, $\gamma\neq \emptyset$, and if $\gamma = \{t_{0}, \ldots, t_{k}\}$, then $\gamma +t =\{t_{0}+t, \ldots, t_{k}+t\}$, $t\geq - t_{0}$.    \par\smallskip
    Using the Walsh expansion \eqref{somma1}    we have
     $f(S^{t}\widehat \zeta) = \sum_{\gamma} f_{\gamma} \widehat \Psi_{\gamma+ t}(\widehat \zeta)$, and we  write
   \begin{equation}\label{split} f(\widehat \zeta) f(S^{t}\widehat \zeta) = f_{\emptyset} f(S^{t}\widehat \zeta) +  C^{(1)}_{t}(\widehat \zeta) + C^{(2)}_{t}(\widehat \zeta) -  R_{t}(\widehat \zeta)\end{equation}
$$ C^{(1)}_{t}(\widehat \zeta) = \sum_{\gamma, \gamma\prime\neq \emptyset\atop M_{\gamma} < m_{\gamma\prime}+t} f_{\gamma} f_{\gamma\prime}  \Psi_{\gamma}(\widehat \zeta) \widehat \Psi_{\gamma\prime +t}(\widehat \zeta), \qquad C^{(2)}_{t}(\widehat \zeta) = \sum_{\gamma, \gamma\prime\neq \emptyset\atop M_{\gamma} \geq m_{\gamma\prime+t} }f_{\gamma} f_{\gamma\prime}   \Psi_{\gamma}(\widehat \zeta)   \Psi_{\gamma\prime +t} (\widehat \zeta) ,$$
and $R_{t}(\widehat \zeta) = \sum_{\gamma, \gamma\prime\neq \emptyset } f_{\gamma}\Psi_{\gamma}(\widehat \zeta) f_{\gamma\prime} \langle \Psi_{\gamma\prime}\rangle \chi(M_{\gamma} \geq m_{\gamma\prime+t})$, where $\chi$ is the indicator function. \par\smallskip
As $|\langle \Psi_{\gamma\prime}\rangle | \leq \|\langle \Psi_{\gamma\prime}\vert \frak  M_{0}\rangle \|_{M}$ we see, by \eqref{lemmadue} and \eqref{correz},  that
\begin{equation}\label{erre} \left \| \left \langle R_{t}(\cdot) \vert \frak M_{0}\right \rangle \right \|_{M} \leq \sum_{\gamma\prime\neq \emptyset} \left |f_{\gamma\prime} \langle \Psi_{\gamma\prime}\rangle \right |  \sum_{\gamma: M_{\gamma}\geq t} |f_{\gamma} \left \| \left \langle  \Psi_{\gamma}\vert \frak M_{0}\right \rangle \right \|_{M} \; \leq \; \cst \; \|\tilde f\|^{2}_{\cC^{\alpha}} \; \kappa^{t} .\end{equation}

Passing to $C^{(1)}_{t}$,
let $\gamma, \gamma\prime\in \frak g$ be such that $M_{\gamma} = r$,  $m_{\gamma\prime} = m$ and $r< m+t$. Taking the conditional expectation with respect to $\frak M_{0}^{r}$ we get by  Proposition \ref{proposition},
$$ \left | \left \langle   \Psi_{\gamma} \widehat \Psi_{\gamma\prime +t} \vert \frak M_{0}\right \rangle \right | = \left | \left \langle  \Psi_{\gamma}\left [\cT^{t+m-r} \widehat G_{\gamma\prime}\right ](\eta_{r}) \vert \frak M_{0}\right \rangle \right | \leq C_{*} \; \mu_{*}^{|\gamma|} \bar \mu^{t+m-r} \|\widehat G_{\gamma\prime} \|_{M}$$
where $\widehat G_{\gamma\prime}$ is defined in \eqref{decompo}. Therefore, again by Inequalities   \eqref{lemmadue} and \eqref{correz}, we see that
 $$ \left \| \left \langle C^{(1)}_{t}(\cdot) \vert \frak M_{0}\right \rangle \right \|_{M}  \leq \cst\; \|\tilde f\|_{\alpha}^{2}
 \sum_{m=0}^{\infty}  \sum_{r=0}^{t+m-1}  \sum_{k=0}^{\infty}\;\bar \mu^{t+m-r}   2^{-\alpha (r+m+k)}   A_{0}^{r} A_{m}^{m+k}  $$
 where, for $0\leq j\leq k\in \Z_{+}$ we set $ A_{j}^{k} : = \sum_{\gamma\in \frak g} \mu_{*}^{|\gamma|} \chi(m_{\gamma}=j, M_{\gamma} =k) \; < (1+\mu_{*})^{k-j+1} $.
 As $\bar \mu < \kappa = 2^{-\alpha}(1+\mu_{*}) <1$,  we get the estimate
  \begin{equation}\label{pezzo1} \left \| \left \langle C^{(1)}_{t}(\cdot) \vert \frak M_{0}\right \rangle \right  \|_{M} \leq \cst\; \|f\|_{\alpha}^{2} \sum_{m=0}^{\infty} 2^{-\alpha m} \sum_{r=0}^{t+m-1} \kappa^{r} \bar \mu^{t+m-r} \leq \cst\; \|f\|_{\alpha}^{2} \;\kappa^{t}. \end{equation}
  \par\smallskip

Turning to  $C^{(2)}_{t}$,   observe that $\Psi_{\gamma} \Psi_{\gamma\prime} = \Psi_{\gamma\Delta \gamma\prime}$, where $\gamma\Delta\gamma\prime = \gamma\setminus \gamma\prime \; \cup \; \gamma\prime \setminus \gamma$,  so that
\begin{equation} \label{quarta}  \left \| \left \langle   C^{(2)}_{t}(\cdot) \vert \frak M_{0}\right \rangle \right \|_{M} \leq \sum_{m, s, k=0}^{\infty}  \;  \sum_{\gamma: M_{\gamma} =t+m+s} |f_{\gamma}| \; \sum_{\gamma\prime:  m_{\gamma\prime} = m \atop  M_{{\gamma\prime}} = m+k} | f_{\gamma\prime}| \left | \left \langle \Psi_{\gamma \Delta \{\gamma\prime+t\}} \vert \frak M_{0} \right \rangle \right |_{M} . \end{equation}
Let $n=  \min \{s, k\}$, $N=\max\{s,k\}$,  $\Omega_{n} = \{m, \ldots, m+n\}$ and  $\gamma_{1}=\gamma \cap \{0, \ldots, m-1\}$, $\gamma_{11}= \gamma\cap \Omega_{n}$,  $\gamma_{12}= \gamma\prime\cap \Omega_{n}$, $\gamma_{2} = \gamma \cup \{\gamma\prime +t\} \cap \{t+n+1,\ldots, t+N\}$. If  $\bar \gamma = \gamma_{11} \Delta \gamma_{12}$  we have
$\gamma \Delta \{\gamma\prime +t\} = \gamma_{1}\cup \bar \gamma \cup \gamma_{2}$,  and the sets $\gamma_{1}, \bar \gamma, \gamma_{2}$ have no common elements.\par\smallskip
 It is not hard to see by induction that
$$ \sum_{\gamma_{11}, \gamma_{12}\subseteq \Omega_{n}} \mu_{*}^{|\gamma_{11}\Delta \gamma_{12}|} = (2(1+\mu_{*}))^{n+1},$$
so that the sum on the right of \eqref{quarta}, for fixed $m,s,k$, is bounded by
$$   \cst\; \|\tilde f\|_{\alpha}^{2} \; \kappa^{t+m} 2^{-\alpha m} \kappa^{N-n}  (2^{-2\alpha+1 } (1+\mu_{*}))^{n} .$$
If $\alpha > 1/2$ and $\bar \mu$ is so small that $2^{-2\alpha+1 } (1+\mu_{*})< 1$, all series converge and we get
\begin{equation}\label{pezzo2} \left \| \left \langle C^{(2)}_{t}(\cdot) \vert \frak M_{0}\right \rangle \|_{M} \right \|  \leq \cst\; \|\tilde f\|_{\alpha}^{2} \;\kappa^{t}. \end{equation}
\par\smallskip Finally, the inequality $|f_{\emptyset}| \| \langle f(S^{t}\cdot
\vert \frak M_{0}\rangle \|_{M} \leq \cst\;  \bar \mu^{t} \|\tilde f\|_{\cC^{\alpha}}^{2}$ is an immediate consequence of Theorem \ref{base} and Lemma \ref{lemma:hoelder}. The proof of \eqref{nnjeravv} follows from this estimate, together with the previous estimates \eqref{erre}, \eqref{pezzo1} and  \eqref{pezzo2}.

 \hfill \square

\subsection*{Acknowledgements}
 A.M. is supported  by the European social fund within the framework of realizing the project ``Support of inter-sectoral mobility and quality enhancement of research teams at Czech Technical University in Prague'' (CZ.1.07/2.3.00/30.0034).\\
 C.S. is supported by ERC Grant MAQD 240518.

\end{document}